\documentclass[11pt]{amsart}
\usepackage{amsmath}
\usepackage{amssymb}
\usepackage{amsfonts}
\usepackage{latexsym}
\usepackage{fancyhdr}
\usepackage{graphics}
\newcommand{\norm}[1]{ \parallel #1 \parallel}
\newcommand{\mb}{\mathbb}
\newcommand{\mc}{\mathcal}
\newcommand{\eul}{\mathfrak}

\newcommand{\Ao}{{\eul A}_{\scriptscriptstyle 0}}

\newcommand{\M}{\eul M}

\newcommand{\HH}{\mc H}

\newcommand{\id}{{\mb I}}

\newcommand{\FF}{\mathfrak F}

\newcommand{\vp}{\varphi}
\newcommand{\vna}{von Neumann algebra}
\newtheorem{defn}{Definition}[section]

\newtheorem{thm}[defn]{Theorem}
\newtheorem{lemma}[defn]{Lemma}
\newtheorem{cor}[defn]{Corollary}
\newtheorem{example}[defn]{Example}
\newtheorem{rem}[defn]{Remark}
\def\x{\relax\ifmmode {\mbox{*}}\else*\fi}
\newcommand{\beex}{\begin{example}$\!\!${\bf }$\;$\rm }
\newcommand{\enex}{ \end{example}}
\newcommand{\berem}{\begin{rem}$\!\!${\bf }$\;$\rm }
\newcommand{\enrem}{ \end{rem}}

\begin{document}
\title[Faithful traces on a von Neumann algebra] {A note on faithful traces\\ on a von Neumann algebra }%
\author{F. Bagarello}%
\address{Dipartimento di Metodi e Modelli Matematici, Universit\`a di Palermo, Facolt\`a d'Ingegneria, I-90128 Palermo (Italy)}%
\email{bagarell@unipa.it}%
\author{C. Trapani}%
\address{Dipartimento di Matematica ed Applicazioni, Universit\`a di Palermo, I-90123 Palermo (Italy)}%
\email{trapani@unipa.it}
\author{S. Triolo}%
\address{Dipartimento di Matematica ed Applicazioni, Universit\`a di Palermo, I-90123 Palermo (Italy)}%
\email{salvo@math.unipa.it}



\begin{abstract} In this short note we give some techniques for constructing,
starting from a {\it sufficient} family $\mc F$ of semifinite or
finite traces on a von Neumann algebra $\M$, a new trace which is
faithful.
\end{abstract}
\maketitle

\section{Introduction and preliminaries}
 It is known that a semifinite von Neumann algebra always has a faithful semifinite trace.
This trace can be used, for instance, to build up a
non-commutative integration and, consequently, to define non
commutative $L^p$-spaces. In this note we give some techniques for
constructing, starting from a family $\FF$ of semifinite traces, a
faithful one which is closely related to the family $\FF$.

  Let ${\FF}=\{\eta_\alpha;\,
\alpha \in {\mc I}\}$ be a family of normal, semifinite traces on
$\M$. We say that the family ${\FF}$ is {\it sufficient} if for $X
\in \M,\, X \geq 0$ and $\eta_\alpha(X)=0$ for every $\alpha \in
{\mc I}$, then $X=0$ (clearly, if $\FF=\{\eta\}$, then $\FF$ is
sufficient if, and only if, $\eta$ is faithful). In this case,
$\M$ is a semifinite von Neumann algebra \cite[ch.5]{takesaki}.
The analysis would really be simplified if, from a given family
$\FF$ of normal semifinite traces, one could extract a sufficient
subfamily ${\mc G}$ of traces with mutually orthogonal supports.
Apart from quite simple situations (for instance when $\FF$ is
finite), we do not know if this is possible or not. There are
however at least two relevant cases where this can be done without
many difficulties. The first case occurs when $\FF$ is countable
and the second  when $\FF$ is a convex and $w^*$-compact family of
finite traces on $\M$. These two situations will be discussed
here.

In the sequel we will need the following Lemmas.

\begin{lemma}\label{thelemma}
Let $\mathfrak{M}$ be a von Neumann algebra in Hilbert space
$\HH$, $\{P_\alpha\}_{\alpha\in\mathcal{I}}$ a
family of projections of $\mathfrak{M}$ with
$$\bigvee_{\alpha\in\mathcal{I}}P_\alpha=\overline{P}.$$
If $A\in\mathfrak{M}$ and $ AP_\alpha=0 $ for every
$\alpha\in\mathcal{I}$, then $A\overline{P}=0.$
\end{lemma}

\begin{proof}
Without loss of generality, we may suppose that
$\mathcal{I}$ is directed upward and so $\{P_\alpha\}$ is a
net.  Were it not so, the family
$$\Gamma:=\{P_{{\alpha}_1}\vee.....\vee P_{{\alpha}_n}; \, {\alpha}_1,........,{\alpha}_n\in\mathcal{I}\}$$
would be a net. Clearly $\sup\Gamma=\overline{P}.$ We would also
have $$A(P_{{\alpha}_1}\vee....\vee P_{{\alpha}_n})=0$$ since
$P_{{\alpha}_1}\vee....\vee P_{{\alpha}_n}$ is the projection onto
the subspace generated by
$P_{{\alpha}_1}\HH,.....,P_{{\alpha}_n}\HH$ and $A$ vanishes on
each one of these subspaces. But, as is known, if $\{P_\alpha\}$
is a net, then $P_\alpha\rightarrow \overline{P}$ strongly. Hence,
for every $\xi\in\HH$,
$$\norm{AP_\alpha\xi-A\overline{P}\xi}\leq\norm{A} \norm{P_\alpha\xi-\overline{P}\xi}\rightarrow 0.$$
So, if $AP_\alpha=0,$ we conclude that $A\overline{P}=0.$
\end{proof}

We remind that if $\varphi$ is a trace on $\M$, the support of
$\varphi$ is the complement of the largest projection of $\M$ that
annihilates $\vp$.

\begin{lemma}\label{lemma2}
Let $\FF=\{\eta_\alpha\}_{\alpha\in\mathcal{I}}$ be a sufficient
family of normal, semifinite traces on the von Neumann algebra
$\mathfrak{M}$ and let $P_\alpha$ be the support of $\eta_\alpha.$
Then, $\vee P_\alpha=\mathbb{I}$, where $\mathbb{I}$ denotes the
identity of $\mathfrak{M}.$
\end{lemma}
\begin{proof}
Indeed, if we put
$S=\mathbb{I}-\bigvee_{\alpha\in\mathcal{I}}\{P_\alpha\},$ we get
$$S=\mathbb{I}-\bigvee_{\alpha\in\mathcal{I}}{P_\alpha}=\bigwedge_{\alpha\in\mathcal{I}}(\mathbb{I}-P_\alpha).$$
Therefore, since $S$ is positive, if $\alpha\in\mathcal{I}$,
$$0\leq \eta_\alpha(S) = \eta_\alpha(\bigwedge_{\beta\in\mathcal{I}}(\id -
P_\beta)) \leq \eta_\alpha(\id - P_\alpha)=0,$$ by definition of
support.\\
Thus, $\eta_\alpha(S)=0$ for every $\alpha \in I$. This implies
that $S=0$ since $\FF$ is sufficient.
\end{proof}
\section{Faithful traces on a von Neumann algebra}

\medskip
Let $\FF =\{\eta_\alpha\}_{\alpha\in\mathcal{I}}$ be a sufficient
family of normal, semifinite traces on the von Neumann algebra
$\mathfrak{M}$. The traces $\eta_\alpha$ are not necessarily
faithful. Let $P_\alpha$ denote the support of $\eta_\alpha$. Then
it is well-known that
\begin{itemize}
  \item[(i)] $P_\alpha\in {\mc Z}(\M)$, the center of $\M$, for
  each $\alpha \in I$.
  \item[(ii)] $\eta_\alpha(X) = \eta_\alpha(XP_\alpha)$, for
  each $\alpha \in
  I$ and for each $X \in \M$.
\end{itemize}
Put $\M_\alpha = \M P_\alpha$. Each $\M_\alpha$ is a \vna\ and
$\vp_\alpha$ is faithful in $\M P_\alpha$ \cite[Proposition V. 2.10]{takesaki}.

More precisely,
$$\M_\alpha:=\M P_\alpha = \{ Z=XP_\alpha, \mbox{ for some } X \in
\M\}.$$

The positive cone $\M_\alpha^+$ of $\M_\alpha$ surely contains the
set

$$\{ Z=XP_\alpha, \mbox{ for some } X \in
\M^+\}.$$ For $Z = XP_\alpha \in \M_\alpha^+$, we put:
$$ \sigma_\alpha (Z):= \eta_\alpha(XP_\alpha).$$
The definition of $\sigma_\alpha (Z)$ does not depend on the
particular choice of $X$. Indeed, if $Z=YP_\alpha$, too, then
$$ \eta_\alpha(XP_\alpha)= \eta_\alpha(ZP_\alpha)= \eta_\alpha(YP_\alpha)$$ since $ZP_\alpha= XP_\alpha = YP_\alpha$.
$\sigma_\alpha$ is a normal, semifinite, faithful trace on
$\M_\alpha$.

\begin{thm}
Let $\FF=\{\eta_n,\,n\in\mathbb{N}\}$ be a countable sufficient
family of normal, semifinite traces on a von Neumann algebra
$\mathfrak{M}$. If $P_n$ denotes the support of $\eta_n,$ then
$$\sigma(X)=\sum_{n\in\mathbb{N}} \eta_n[P_n
\prod_{k<n}(\mathbb{I}-P_k)X] \quad \quad X\in\mathfrak{M}^+$$ is
a faithful semifinite normal trace on $\mathfrak{M}.$
 \end{thm}

\begin{proof}
We define
\begin{eqnarray*}Q_1&=& P_1
 \\  Q_n&=&P_n \prod_{k<n}(\mathbb{I}-P_k).
\end{eqnarray*}
If $n \neq m$, then
$$Q_nQ_m=P_n\prod_{k<n}(\mathbb{I}-P_k)P_m \prod_{h<m}(\mathbb{I}-P_h)=0.$$

Therefore the $Q_n$'s are orthogonal. It is clear that they are
idempotent.

We now prove that $S=\mathbb{I}-\sum_{n=1}^{+\infty}Q_n=0,$ where
the limit is taken in the strong operator topology.
\\If
$S_m=\mathbb{I}-\sum_{n=1}^{m}Q_n$ then
$$S_m=\prod_{n=1}^{m}(\mathbb{I}-P_n).$$
In fact, by induction, we have:

$$S_1=\mathbb{I}-Q_1=\mathbb{I}-P_1$$
\begin{eqnarray*}
 S_{m+1}&=&\mathbb{I}-\sum_{n=1}^{m+1}Q_n =\mathbb{I}-\sum_{n=1}^{m}Q_n-Q_{m+1}\\&=&
\prod_{n=1}^{m}(\mathbb{I}-P_n)-Q_{m+1}
=\prod_{n=1}^{m}(\mathbb{I}-P_n)-P_{m+1}
\prod_{n=1}^{m}(\mathbb{I}-P_n) \\ &=& (\mathbb{I}-P_{m+1})
\prod_{n=1}^{m}(\mathbb{I}-P_n) =\prod_{n=1}^{m+1}(\mathbb{I}-P_n)
\end{eqnarray*}
then $$ S_mP_l=0 \quad \mbox{ if } \quad  l \leq m.$$ Letting
$m\rightarrow+\infty,$  we get $SP_l=0$, for every
$l\in\mathbb{N}.$ By Lemma \ref{thelemma}, $S=0.$\\ Now, we define
$$\sigma_n(X)=\eta_n(Q_nX)\quad \forall n\in\mathbb{N}.
$$Then $\sigma_n$ is a semifinite normal trace with support $Q_n.$
Indeed, let $R$ be a projection with $\sigma_n(R)=0.$ Then
$$\sigma_n(R)=\eta_n(Q_nR)=0 \Rightarrow Q_n R \leq \mathbb{I} -
P_n \Rightarrow Q_n R(\mathbb{I} - P_n) = Q_n R \Rightarrow Q_n R
P_n=0.$$ But, since the $P_n$'s are in the center of
$\mathfrak{M}$,
$$Q_n RP_n=R P_n \prod_{k<n}( \mathbb{I}-P_k )P_n=RQ_n=0 $$
then $$R \leq \mathbb{I}-Q_n,$$ which implies that $Q_n$ is the
support of $\sigma_n.$

 Thus the function
$\sigma$ on $\mathfrak{M}^+$ defined by
$$\sigma(X)=\sum_{n\in\mathbb{N}}  \sigma_n(X) \quad \quad X\in\mathfrak{M}^+$$ is a semifinite
normal trace whose support is $\sum_{n \in {\mathbb
N}}Q_n=\mathbb{I}$ \cite[ch.5 lemma 2.12]{takesaki}. Therefore,
$\sigma$ is  faithful on $\mathfrak{M}.$

\end{proof}

We now try to remove the assumption that $\FF$ is countable. As we
shall see, some alternative hypothesis should be made.

The following Lemma has been proved in \cite{btt}. We give a
sketch of the proof for the sake of completeness.
\begin{lemma} \label{LEMMA3}
Let $\FF$ be a convex $w^{\ast}$-compact family of normal, finite
traces on a von Neumann algebra $\mathfrak{M}$; assume that, for
each central operator $Z$ with $0\leq Z\leq {\mb I}$, and each
$\eta\in\FF$ the functional $\eta_{\scriptscriptstyle Z}$, defined
by $\eta_{\scriptscriptstyle Z}(X):=\eta(XZ), \; X \in \M,$ still
belongs to $\FF$. Let $\mathfrak{E}\FF$ be the set of extreme
elements of $\FF$. If $\eta_1, \eta_2\in\mathfrak{E}\FF$, $\eta_1
\neq n_2$, and $P_1$ and $P_2$ are their respective supports, then
$P_1$ and $P_2$ are orthogonal.
\end{lemma}

\begin{proof}
Let $P_1, P_2$ be, respectively, the supports of $\eta_1$ and $\eta_2$. We begin with proving that either $P_1=P_2$ or $P_1 P_2=0$. Indeed, assume that
 $P_1P_2\neq 0$. We define
$$\eta_{1,2}(X)=\eta_1(XP_2) \quad \quad x \in \M.$$
Were $\eta_{1,2}=0$, then, $\eta_{1}(P_2)=0$ and therefore
$P_1P_2=0$, which contradicts the assumption. It is easy to see
that the support of $\eta_{1,2}$ is $P_1P_2.$

Thus $\eta_{1}$ majorizes  $\eta_{1,2}.$ But $\eta_{1}$ is extreme
in $\FF.$ Then $\eta_{1,2}=\lambda\eta_{1}$ for some
$\lambda\in[0,1]$. This implies that $\eta_{1,2}$ has the same
support as $\eta_{1}$; therefore $ P_1P_2=P_1$ i.e. $ P_1 \leq
P_2.$ Starting from $\eta_{2,1}(X)=\eta_2(XP_1)$, we get, in
similar way, $ P_2 \leq P_1.$ Therefore, $P_1P_2 \neq 0$ implies
$P_1=P_2$.

However,  two different traces of $\mathfrak{E}\FF$ cannot have
the same support. Indeed, assume that there exist $\eta_1,
\eta_2\in\mathfrak{E}\FF$ having the same support $P.$ Since $P$
is central, we can consider the von Neumann algebra
$\mathfrak{M}P$. The restrictions of $\eta_1, \eta_2$ to
$\mathfrak{M}P$ are normal faithful finite traces. By \cite[ch.5.
2.31]{takesaki} there exist a central element $Z$ in
$\mathfrak{M}P$ with $0\leq Z \leq P$ ($P$ is here considered as
the unit of $\M P$) such that
\begin{equation} \label{uno}\eta_1(X)=(\eta_1+\eta_2)(ZX) \quad \forall
X\in\mathfrak({\M}P)_+.\end{equation} The operator $Z$ belongs to
the center of $\mathfrak{M}$. Therefore the functionals
$$\eta_{\scriptscriptstyle {1,Z}}(X):=\eta_1(XZ) \quad \quad \eta_{\scriptscriptstyle {2,Z}}(X):=\eta_2(XZ) \quad\quad  X\in\mathfrak{M}$$
belong to the family $\FF$ and are majorized, respectively, by the
extreme elements $\eta_1,\eta_2.$ Then, there exist $\lambda \in
[0,1[$ and $\mu \in]0,1]$ such that
$$\eta_1(XZ)=\lambda\eta_1(X)\quad \quad \eta_2(XZ)=\mu\eta_1(X), \quad \forall X \in \M.$$
>From the equalities, it follows, for instance, that either
$\eta_1$ is a convex combination of $\eta_2$ and $0$ or $\eta_2$
is a convex combination of $\eta_1$ and $0$. This is absurd.
\end{proof}
\berem It is worth noticing that the assumptions of the previous
Lemma are satisfied when $\FF$ is the family of all traces $\eta$
on $\mathfrak{M}$ such that $\|\eta\|=1$. \enrem
\begin{lemma}\label{lemma9}
Let $\FF$ be a convex $w^{\ast}$-compact family of positive linear
functionals on a  $C^{\ast}$-algebra $\Ao$ and let
$\mathfrak{E}\FF$ the set of extreme elements of $\FF.$ We
have:$$\sup_{\eta\in \FF}\eta(a^{\ast}a)=\sup_{\eta\in
\mathfrak{EF}}\eta(a^{\ast}a) \quad\quad \forall a\in\Ao.$$
\end{lemma}

\begin{proof}
It is clear that
$$\sup_{\eta\in
\FF}\eta(a^{\ast}a)\geq\sup_{\eta\in
\mathfrak{E}\FF}\eta(a^{\ast}a) \quad\quad \forall
a\in\mathfrak{A}_0.$$ But every $\eta\in \mathfrak{F}$ is in
$w^{\ast}$-closure of the convex hull of $\mathfrak{E}\FF$, thus
$$\forall a\in\mathfrak{A}_0 \quad \forall \epsilon>0 \quad
\exists \eta_1,\eta_2......\eta_n\in\mathfrak{E}\FF\quad \quad
\lambda_1,\lambda_2......\lambda_n\in\mathbb{R}^+\cup 0 \,:
\sum_{i=1}^ {n} \ \lambda_i = 1 $$

such that,
$$\eta(a^{\ast}a)< \sum_{i=1}^ {n} \ \lambda_i \eta_i(a^{\ast}a) + \varepsilon \leq
\sum_{i=1}^{n} \ \lambda_i \sup_{\eta\in
\mathfrak{E}\FF}\eta(a^{\ast}a) + \varepsilon =\sup_{\eta\in
\mathfrak{E}\FF}\eta(a^{\ast}a)+\epsilon.
$$
Then
$$\sup_{\eta\in
\FF}\eta(a^{\ast}a)\leq \sup_{\eta\in
\mathfrak{E}\FF}\eta(a^{\ast}a) \quad\quad \forall
a\in\mathfrak{A}_0.$$
\end{proof}

\begin{thm}\label{Teorema}
Let $\FF$ be a convex $w^{\ast}$-compact sufficient family of
normal, finite traces on a von Neumann algebra $\mathfrak{M}$;
assume that, for each central operator $Z$, with $0\leq Z\leq {\mb
I}$, and each $\eta\in\FF$ the functional
$\eta_{\scriptscriptstyle Z}(X):=\eta(XZ)$ belongs to $\FF$. Let
$\mathfrak{E}\FF$ be the set of extreme elements of $\FF$. Then
the function $\sigma$ on $\mathfrak{M}^+$ given by
$$\sigma(X)=\sum_{\eta\in \mathfrak{E}\FF}  \eta(X) \quad \quad
X\in\mathfrak{M}^+$$ is a faithful semifinite normal trace on
$\mathfrak{M}.$
 \end{thm}
 \begin{proof}
By Lemma \ref{LEMMA3} the function $\sigma$ on $\mathfrak{M}^+$
given by
$$\sigma(X)=\sum_{\eta \in \mathfrak{E}\FF} \eta (X) \quad \quad X\in\mathfrak{M}^+$$ is
a semifinite normal trace. We prove that $\sigma$ is faithful. By
Lemma \ref{lemma9}, $\mathfrak{E}\FF$ is sufficient, then by Lemma
\ref{lemma2}, $\vee P_\eta=\mathbb{I}$, where $\mathbb{I}$ denotes
the identity of $\mathfrak{M}$ and $P_\eta$ the support of the
${{\eta_i}\in \mathfrak{E}\FF}.$  It is clear that $\mathbb{I}$ is
the support of $\sigma$ thus $\sigma$ is a faithful semifinite
normal trace on $\mathfrak{M}.$

\end{proof}

Let ${\FF}=\{\eta_\alpha;\, \alpha \in {\mc I}\}$ be a sufficient
family of normal, finite traces on $\M.$ In order to get a similar
result in the case where $\FF$ is not necessarily convex or
$w^{\ast}$-compact, it is enough to assume that the family $\FF$
is uniformly bounded, i.e. $\eta_\alpha(\mathbb{I})\leq 1 $, for
every $ \alpha\in {\mc I}$.

We put
$$co\{{\FF}\}=\left\{\sum_{i=1}^{n}\lambda_i \eta_{\alpha_i} \,;\quad \quad
\lambda_i\geq 0, \quad \sum_{i=1}^{n}\lambda_i=1, \quad\quad
\eta_{\alpha_i}\in \FF \right\} $$ and let
$\overline{co\{{\FF}\}}$ be its $w^*-$closure.

\begin{cor} Let ${\FF}=\{\eta_\alpha;\, \alpha \in {\mc I}\}$ be a sufficient
family of normal, finite traces on $\M$ such that
$\eta_\alpha(\mathbb{I})\leq 1 $, for every $ \alpha\in {\mc I}$.
Let $\mathfrak{E}\overline{co\{{\FF}\}}$ be the set of all extreme
elements of $\overline{co\{{\FF}\}}.$ Then the function $\sigma$
on $\mathfrak{M}^+$ given by
$$\sigma(X)=\sum_{\eta
\in\mathfrak{E}\overline{co\{{\FF}\}}}  \eta (X) \quad \quad
X\in\mathfrak{M}^+$$ is a faithful semifinite normal trace on
$\mathfrak{M}.$

\end{cor}
\begin{proof}
By the assumption, $\FF$ is a subset of the unit ball of the dual
of $\mathfrak{M}.$ Then $\overline{co\{{\FF}\}}$ is a convex
$w^{\ast}$-compact subset of the dual of $\mathfrak{M}.$ It is
easily seen that the elements of $\overline{co\{{\FF}\}}$ are
traces. Indeed every $\sum_{i=1}^{n}\lambda_i \eta_{\alpha_i}$ is
a trace since
$$\sum_{i=1}^{n}\lambda_i
\eta_{\alpha_i}(x^*x)=\sum_{i=1}^{n}\lambda_i
\eta_{\alpha_i}(xx^*)$$ and if $\{\eta_\gamma\}$ is a net of
traces and $\eta=w^*-\lim_\gamma \eta_\gamma,$ we have:

$$\eta(x^*x)=\lim_\gamma \eta_\gamma(x^*x)=\lim_\gamma \eta_\gamma(xx^*)=\eta(xx^*) \quad x\in\mathfrak{M}.$$
By Theorem \ref{Teorema} the function $\sigma$ on $\mathfrak{M}^+$
given by
$$\sigma(X)=\sum_{\eta
\in\mathfrak{E}\overline{co\{{\FF}\}}}  \eta (X) \quad \quad
X\in\mathfrak{M}^+$$ is a faithful semifinite normal trace on
$\mathfrak{M}.$
\end{proof}
\bibliographystyle{amsplain}

\end{document}